\documentclass[11pt]{article}
\usepackage{amssymb}
\usepackage{mathrsfs}
\usepackage{graphics}
\usepackage{amsmath}
\usepackage{amsmath,amsthm}
\usepackage{fullpage}
\usepackage[longnamesfirst]{natbib}
\usepackage{setspace}
\usepackage[demo]{graphicx}
\usepackage{soul}
 \usepackage{url} \usepackage{booktabs} 
\usepackage[linesnumbered,ruled,noend]{algorithm2e} 

\newtheorem{theorem}{Theorem}
\newtheorem{corollary}{Corollary}
\newtheorem{lemma}{Lemma}

\newtheorem{proposition}{Proposition}

\newtheorem{definition}{Definition}
\newtheorem{remark}{Remark}

\newtheorem{example}{Example}

\begin{document}

\title{Zero Carbon  V2X Tariffs  for Non-Domestic Customers}
%
\author{Elisheva Shamash\thanks{Technion -- Israel Institute of Technology. Email: elisheva.shamash@gmail.com.} \and Zhong Fan\thanks{University of Exeter. Email: Z.Fan@exeter.ac.uk.}  
}

\maketitle

\begin{abstract}

\vspace{\baselineskip}

With the   aim of meeting the worlds net-zero objectives, electricity trading through contractual agreements  is becoming increasingly relevant in global and local energy markets. 
We develop  contracts  enabling efficient energy trading using \textit{Vehicle to Everything (V2X)} technology which can be applied to regulate  energy markets and 
reduce costs and carbon emissions by using   electric vehicles with bi-directional batteries to
store energy during off-peak hours for export during peak hours.
We introduce a  contract   based on the VCG mechanism [Clarke, 1971] which enables fleets of electric vehicles to  export electricity to the grid efficiently throughout the day, where each electric vehicle has its  energy consumption schedules and costs and  it's energy exporting contract  schedules.

\end{abstract}


\begin{titlepage}
\maketitle
\end{titlepage}

\section*{Nomenclature - General}

\begin{itemize}
\item \textit{EV}:  an electric vehicle with a bi-directional battery that can import (export) energy

    \item $kW$: Kilowatt
    \item  $kWh$: Killowatt-hour 
    \item \textit{platform}: A computerized platform that allocates energy transaction contracts to EVs 
    \item $hhp$: half hour periods
in which energy can be exported or imported

   \item \textit{V2X Technology (Vehicle-to-Everything)}: the transfer of the electricity stored in EVs  to  grids (V2G),  to buildings (V2B), to loads (V2L)  etc.

   \item $\mathbb{T}$: a 24 hour day constructed by $hhp$s

\item $Peak$: A number of consecutive $hhp$s of high electricity demand during $\mathbb{T}$

\item $Valley$, $off-peak$: A number of consecutive $hhp$s of low electricity demand during $\mathbb{T}$

\item $D=(D^{hhp_1},D^{hhp_2}, ..., D^{|\mathbb{T}|} ) $:  vector of expected energy demand throughout the day

\end{itemize}

\vspace{2mm}

 \textbf{The Hour-Scheduling mechanism Nomenclature}

\vspace{2mm}
\begin{itemize}

\item  \textit{Hour-Scheduling mechanism}:  enables the platform to allocate energy exporting contracts 

\item \textit{safety margin ($sm$)}:  added to the  energy demand vector $D$ in order to increase energy reliability

\item $\hat{D}^{hhp}$: expected energy demand with a safety margin, i.e.,   $\forall hhp\in \mathbb{T}:\hat{D}^{hhp}=D^{hhp}+sm$

\item \textit{Spot market price $m^{hhp}_i\in M_i$}: where $M_i $ is the vector of Spot market prices  $\forall hhp\in \mathbb{T}$ at market  $i\in\{day-ahead,intra-day,balancing\}$

\item \textit{Battery Usage Deterioration Cost ($c^{BD}$)}:     battery deterioration expected cost for charging and discharging  1 $kW$

\item \textit{bundle ($\omega$)} : Each  EV's imported  set of $kW$ is termed a \textit{bundle} and is denoted as  $\omega$

\item  $\mathbb{J}$:  the set of 
 contracts that can be offered in the Hour-Scheduling mechanism 

\item  $j=(n, b, \tilde{hhp},\omega,f_l)\in \mathbb{J}$: 
a contract offered by  fleet $n$ in the Hour-Scheduling mechanism, 
 where   $b$ is  fleets $n$'s announced bid for contract $j$, $\tilde{hhp}$ is the $hhp$ in which contract $j$ should be implemented, $\omega$ is the bundle which is to be exported in contract $j$ and $f_{\ell}$ is the fine if fleet $n$ does not honor the contract and $\omega$ was not exported

\item $J^{offered}\subseteq \mathbb{J}$ : the set of contracts   offered   in the Hour-Scheduling mechanism

\item $J^{offered}_n\subseteq J^{offered}$ : the set of contracts fleet $n\in N$ offers  in the Hour-Scheduling mechanism 

\item $WB$ : The set of all distinct bundles $\omega$ offered by EVs   via contracts  $J^{offered}$ 

\item $n:\mathbb{J}\rightarrow \mathbb{N}$:  the fleet $n$ that offered contract  $j\in \mathbb{J}$

 \item $f_{\ell}:\mathbb{J}\rightarrow \mathbb{R}_{\geq 0}$: 
  the fine fleet $n(j)$ pays if they do not honor contract $j\in \mathbb{J}$ 

\item $\ell:\mathbb{J}\rightarrow \mathbb{R}_{\geq 0}$:  the amount of $kWh$ offered in contract $j\in \mathbb{J}$

\item $b:\mathbb{J}\rightarrow \mathbb{R}_{\geq 0}$:  the bid offered for contract $j\in \mathbb{J}$

\item $\omega:\mathbb{J}\rightarrow WB$:  the bundle of $kWh$ offered in  contract  $j\in \mathbb{J}$

\item $J^{\omega}$ :  the set of contracts  offering to export bundle $\omega\in WB$, i.e. $\forall \omega'
\in WB: J^{\omega'}=\{ j\in \mathbb{J}|\omega(j)=\omega' \}$

\item  $WB_k$ : the $kth$ bundle in vector of bundles $WB$, such that $\forall k\in  [1,|WB|]$

\item  $\overline{WB}(J^{offered})$: the vector of sets of offered contracts $\omega\in WB$. Also abbreviated as $\overline{WB} $ when clear from context, where $\overline{WB}=\{J^{WB_1},..., J^{WB_{|WB|}}\} $

\item $\overline{WB}_k$: the set of contracts offering to export bundle $WB_k$, i.e. $\overline{WB}_k=J^{WB_k}$

\item $\hat{J}(\overline{WB},\hat{D},M)\subseteq J^{offered}$: contracts the platform  returns via the Hour-Scheduling mechanism

\item $\breve{{WB}}^q=(\omega_1, \omega_2,..., \omega_{q})\subseteq WB$:  the vector of bundles in  ${WB}$  proceeding and including  bundle $\omega_q$, i.e. $(WB_1,...,WB_{q-1}, WB_q)$

\item $\breve{\overline{WB}}^q=(J^{\omega_1},...,J^{\omega_q})$:  the vector of sets of contracts offering bundles $\omega\in WB$ preceding and including the set of contracts  $J^{\omega_q}=\overline{WB}_q$ 

\item $SocSave({J},D,M)$:  returns society's savings (i.e. social welfare) calculated in  algorithm~\ref{algValue}

\item $p:\mathbb{J}\rightarrow [0,1]$: denotes  the private probability of fleet $n(j)$ fulfilling contract $j\in \mathbb{J}$

\item $u^n(\overline{WB},\hat{D},M)$: fleet $n$'s utility   by engaging with the Hour-Scheduling mechanism, given $\overline{WB}$, $\hat{D}$ and $M$

\item $u^{system}(\overline{WB},\hat{D},M)$: the platform's  utility 
 by engaging with the Hour-Scheduling mechanism given $\overline{WB}$, $\hat{D}$ and $M$

 \item $payment^n(\overline{WB},\hat{D}, M)$:  the payment the platform transfers to fleet $n$ due to fleet $n$'s 
  engagement with the Hour-Scheduling mechanism 

\item $c(j)$: the cost of fleet $n(j)$ attempting to implement contract $j$ 

\item $\Ddot{\Ddot{J}}$ and $\Ddot{\Ddot{SocSave}}$: denote tables  of sizes $(|WB|+1) \times |\mathbb{T}| \times 1$, where    table $\Ddot{\Ddot{J}}$  describes  potential sets of contracts $J(\cdot)$ given $WB'\subseteq WB, \ d\leq D, \ and \ M$, and  table $\Ddot{\Ddot{SocSave}}$ describes the respective society's savings


\end{itemize}

\vspace{2mm}
\textbf{The Frequency Regulating mechanism (FR) Nomenclature}
\vspace{2mm}

\begin{itemize}

\item \textit{Ancillary Service
market}: A market facilitating  reliability  of energy in real time

\item \textit{Frequency Regulating (FR) mechanism}: enables  EVs'    participation in  Ancillary Services  via charging and/or discharging energy to/from the platform

\item  \textit{State of Charge ($SoC$)}:   the number of $kW$ a battery is charged

    \item $\forall hhp\in \mathbb{T}:SoC^m(hhp)$:  the $SoC$ of EV $m$'s  battery at $hhp\in \mathbb{T}$ 
    
    \item $plugged(hhp)$: the set of EVs connected  to the platform at $hhp\in \mathbb{T}$
    
    \item $x^m_{max}(hhp)$ ($x^m_{min}(hhp)$):  the desirable maximal (minimal)  $SoC$ in  EV $m$'s battery at the end of  
    $hhp\in \mathbb{T}$, based on EV $m$'s schedule

\item   $J^{active}$: 
  the set of accepted and active  energy transacting contracts that have not yet defaulted
  
    \item  ${p}_{hhp}(y,J^{active}$):   the   probability that  $y$ $kWh$ will   be supplied given contracts 
$J^{active}$ at $hhp\in \mathbb{T}$

\item $X_{hhp}^{max}(J^{active})$:  denotes the maximal  exported energy at $hhp\in \mathbb{T}$ if all contracts $J^{active}$ are honored

\item $payment^m_{FR}(\cdot)$:  EV $m$'s payment for  connecting to the platform via the Frequency Regulating mechanism

\end{itemize}

\section{Introduction}
\label{sec:intro}

Electricity markets are responsible for a large portion of carbon
emissions, therefore improving  electricity market methods for  energy
 and transactions   is expected to have  a significant impact on carbon emission reduction and thus
facilitate the transition to Net Zero~\footnote{See \url{https://www.nationalgrid.com/stories/energy-explained/what-is-net-zero}}. 
Specifically, during periods of high demand, back-up power plants are often
needed to supply electricity demands. These back-up power plants are usually
fossil-fuel powered generators which produce higher carbon emissions than most
other types of power generation facilities. Therefore, during peak hours, additional electricity is sometimes needed to
be purchased from the wholesale electricity market  raising electricity costs. For example,
 the top 1\% of electricity demand accounts for 8\% of electric energy
costs\footnote{See \url{https://www.hged.com/smart-energy/what-is-peak-electricity-demand.aspx}}. 

In order to reduce carbon emissions,
renewable energy technology is developing and expanding, including considerable increases in offshore wind and  in solar.
To integrate this low carbon technology, the energy system  requires a huge increase in flexibility provided by energy storage, interconnectors, and demand-side response, to  
50GW by 2030\footnote{See \url{https://assets.publishing.service.gov.uk/media/64a54c674dd8b3000f7fa4c9/offsho}}. The expansion in low-carbon technology alongside supply and demand flexibility  will ensure that the future Net Zero energy system has sufficient supply during challenging periods. Studies show that a flexible grid could save the UK £10-17 billion per year by 2050 by reducing the need for dedicated back-up generation and grid reinforcement requirements to meet peak demand.\footnote{See \url{https://www.energy-uk.org.uk/insights/balancing-the-system/}.}

In pursuance of developing this required flexibility, many  such technologies are developing.  
The energy management \textit{V2X (Vehicle-to-Everything)} is such a developing technology,   where V2X includes the transfer of the electricity, stored in electric vehicle (EV) batteries, vehicle to the grid (V2G), vehicle to buildings (V2B) and vehicle to houses (V2H)\footnote{see \url{https://www.dnv.com/energy-transition-outlook/v2x/}}.
V2X technology can add flexibility to energy markets, and thus
reduce costs and carbon emissions by using  electric vehicle batteries to
jointly store energy during off-peak hours and export it during peak hours.
 Specifically, V2X can generate revenue by performing energy arbitrage by
discharging (i.e. exporting energy) during expensive peak hours and charging (i.e.
importing energy) during off-peak hours.

There are 32,992 public charging points across the UK (May 2024), which is an increase of 43\% in the number of public charging devices since 2023 
(ZapMap). 
With 50,000 V2G-enabled EVs, each EV could reduce system operation costs by approximately £12,000 per annum and CO2 emissions by around 60 tonnes per annum (e4Future project)\footnote{\url{https://www.chademo.com/economic-benefits-and-co2-saving-to-the-electricity-system-quantified-in-a-new-v2g-study-in-the-uk}.}. 
Therefore, V2X technology is extremely suitable for  carbon emission reduction.

The energy market consists of three main
markets\footnote{\url{https://sympower.net/what-are-energy-markets/} 
} where participants trade  24 hours a day. 
\begin{itemize}
\item the Forward market, where energy transactions are settled weeks, months or years in
advance;
\item the Spot market, where energy transactions are settled  a day ahead, or up to 5
minutes in advance;
\item the Balancing Market 
(Ancillary Service)\footnote{\url{https://www.arnepy.com/balancing-market-ancillary-service}}  for facilitating frequency regulation and easing the balancing market in real time. 
\end{itemize}
Energy transactions are based on contracts between an energy grid and energy importers  (i.e., energy suppliers and consumers). By applying smart contracts, 
 V2X technology can reduce carbon emissions and costs stemming from the above
three markets. We shall focus on the day-ahead Spot market.

For example, consider a micro-grid from which non-domestic customers
including  fleets of EVs, import energy for their use. Fleets often
have an excess quantity of $kWh$ in their EV batteries. By means of a \textit{computerized platform (platform)} 
fleets with bi-directional batteries can export excess energy back to the micro-grid. The platform allocates  contracts determining transactions between the fleets and the micro-grid, stating the quantities of $kWh$ that will be exported from the fleets' EVs to the
micro-grid,  the corresponding payment each fleet will receive
and the time of the transactions. The platform determines these transactions in accordance with the corresponding Spot markets via 
 an algorithmic procedure.

\subsection{Literature Review}

Many proposed mechanisms for 
energy trading between EVs  and energy aggregating platforms
via V2X rely on
machine learning or statistical analysis by aiming to assess the demand and
supply of electricity throughout the day and offer prices for importing and
exporting electricity accordingly (see~\cite{inci2022integrating, tepe2022optimal, liu2013opportunities, signer2024modeling, zheng2020day}). These
techniques require extensive data collection, which might not always be available.
Furthermore, companies and the public, who are strategic agents, might offer untruthful  reports in order to raise their profits, 
or might not  share their  private information  in order to enhance their final profits,
 possibly  resulting in extreme in-efficiency (see~\cite{nisan2007computationally, alon2023incomplete, osogami2023learning}).

Further works include \cite{hou2021simultaneous} who present a multi-round auction for EV charging and introduce a dynamic charging scheduling algorithm, and 
\cite{luo2021blockchain, gao2021price}  designing a double auction mechanism for energy scheduling management, where discharging EVs trade energy with the grid or EVs, and a price adjustment strategy for  is proposed.
Other proposed mechanisms consider blockchain energy markets~\cite{sagawa2021bidding, liu2019electric}) and multi-stage algorithms~\cite{zeng2015incentivized, kanakadhurga2022peer}.

According to \cite{Holmstrom79} and \cite{alon2023incomplete}, the only mechanisms that
can induce players to announce their true valuations, and thus enable  efficient outcomes
are VCG-based mechanisms (i.e., \cite{Clarke71, Vickrey61, Groves73}),
which by nature, induce participates to be truthful by having fleets' payments equal to their \textit{opportunity cost}  \cite{nisan2007computationally}.
Indeed, there exist many works considering VCG based mechanisms for trading energy including \cite{umoren2023vcg, ma2021incentive, yassine2019double, zhong2020multi, zhang2018online, meir2017contract}.
For example, 
\cite{umoren2023vcg, zhong2018topology} 
 consider electricity trading via VCG-based auctions, where the trade takes into account the distance between the EV and the  platform. 
Although many papers consider applying VCG-based mechanisms to V2G, these papers do not take into account  interrelated effects on participants profits from the different time periods  for trading energy (by trading energy during a particular $hhp$ an agent may have to forgo profits from a different $hhp$), nor do they consider non-domestic customers such as fleets of EVs (where each EV has it's own schedule and costs), rather than individual EVs.

\subsection{Main contributions}

The current understanding is that there exists no system in the market that
optimizes V2X mechanisms for electricity trading by non-domestic customers.

In this research, we focus mainly on the Spot market and on the Frequency regulating
market.  
We believe that the Forward market 
may be developed further in a similar fashion to our suggested contract  corresponding to the Spot market, but leave this open for future research. 
Particularly, we consider fleets of
Electric Vehicles exporting energy via a computerized system to  a 
Micro-Grid
via energy exporting contracts. We refer
to the computerized system as the \textit{platform}. For simplicity, we 
consider this platform as the energy recipient, as well as the system operator for allocating contracts and transferring (or receiving) payments.

This paper introduces a novel application of a VCG-based mechanism~\cite{Clarke71} for electricity trading,
for the  day-ahead spot market using smart contracts in a V2X setting that is 
efficient, cost minimizing and simple to implement. We term this application the 
 \textit{Hour-Scheduling mechanism}. The Hour-Scheduling mechanism  allocates electricity exporting contracts to participating EVs alongside  the  Spot markets. 

This mechanism allows a set of $N$ fleets, each consisting of a number of EVs, to collectively store energy during off-peak hours for export during peak hours and thus reduce carbon emissions,  energy costs and  curtailment\footnote{See \url{https://www.next-kraftwerke.com/knowledge/curtailment-electricity}.}.
Specifically, the mechanism enables the $N$ fleets to engage in contracts  stating  quantities of $kWh$ they wish to export at  pre-specified times to the platform,  thus enabling the system to divert importing $kWh$ from the grid during peak hours to off-peak hours.

Although many previous works consider energy trading via VCG mechanisms, they do not take into account possible manipulations by agents applied in order to export their energy at more profitable hours, as well as the added complexities of considering fleets of electric vehicles rather than independent EVs. 

Our approach takes into account the effect of different \textit{half hour periods ($hhp$)}  for trading energy and their interrelated effects on fleets with multiple EVs participating. This approach can
simultaneously induce truthfulness, raise energy trading efficiency and 
can be computed in polynomial time (see theorem~\ref{th12032024})\footnote{In essence, this setting is a specific case of   combinatorial auctions, and may be extended further to a wide range of combinatorial problems~\cite{myers2023learning}.}.
Furthermore, lower bounds for the probability of covering the expected demand  can be deduced from the announced fleet bids for offered
contracts (see proposition~\ref{propBoundForProbabilty}). 

Determining the optimal set of energy transactions is an NP-hard problem by a reduction from the Knapsack problem~\cite{karp1972complexity}. 
Following~\cite{meir2017contract} we solve this problem using dynamic programming, which we extend  to a multi-dimensional case, thus allowing for multiple  $hhp$'s in a 24 hour period, and for fleets offering multiple contracts for exporting the same energy for each $hhp$ via multiple EVs.

Moreover, we introduce a \textit{Frequency Regulating}     
mechanism which enables frequency regulation via EV's $kWh$ charging and discharging,  thus supporting the Ancillary Services market. 
At every given moment, the difference between demand and supply must remain within a certain limit
\footnote{\url{https://www.neso.energy/industry-information/balancing-services/network-services-procurement/voltage-network-services-procurement}}, otherwise black outs may occur. The  Ancillary Services market's purpose is to maintain this difference within its' bounds. Thus the Frequency Regulating mechanism enables fleets to earn additional profit from energy arbitrage by connecting to the platform and importing or exporting energy, in accordance with their batteries' capacities and their  energy requirements, and thus facilitate carbon emission reduction and costs.

The Hour-Scheduling mechanism and the Frequency Regulating mechanism can be  applied  for both domestic and non-domestic use. Specifically,
these mechanisms may offer EV or EV fleets a significant additional source of income, enabling their participation in the energy market, while reducing energy costs. \cite{SothernwoodV2X} describes a possible  transition of electrifying busses and supporting their participation in the Hour-Scheduling and the Frequency regulating mechanisms. 

In this paper, the Hour-Scheduling and Frequency Regulating mechanisms are described in a setting with a micro-grid with which non-domestic customers, consisting of fleets of EVs, import and export energy. 
Our model aims to minimize energy costs for the platform which we assume also  minimizes carbon emissions\footnote{See \url{https://www.hged.com/default.aspx}}. 

In section~\ref{section2} we present our general setting.
Section~\ref{sectionHourSchedulingMec} introduces the  setting  of  the Hour-Scheduling mechanisms for the day-ahead spot market,  and section~
\ref{secHourSchedulingMechanism} introduces the Hour-Scheduling mechanism's contract.
Section~\ref{secAnalysis} presents an analysis  of the Hour-Scheduling mechanism. 
In section~\ref{secFreq} we present the Frequency Regulating mechanism, and section~\ref{secConclusion} concludes.

\section{The General Setting}\label{section2}

We develop two mechanisms for energy trading between fleets of EVs and a local grid termed  
the {Hour-Scheduling mechanism} and the 
{Frequency Regulating mechanism}.

The 24 hour day consists of 48 \textit{half hour periods (hpp)}. We denote $\mathbb{T}$  as the set of all $hhp$'s  during a 24-hour day, such that   $\mathbb{T}=\{$00:00-00:30,  00:30-01:00,  01:00-01:30,    ... , 23:00-23:30,  23:30-00:00$\}$.

The set $\mathbb{T}$  consists of  two subsets:   
$$Peaks=\{peak_1,...,peak_{|Peaks|}  \} \ \   and \  \ Valleys=\{valley_1,...,valley_{|Valleys|}  \},\footnote{In this paper, we refer to the terms \textit{valley} and \textit{off-peak} interchangeably.}$$ where 
between every  two consecutive peaks: $peak_i, peak_{i+1}\in Peaks$ there  is a $valley\in Valleys$ and vise versa. 
Every peak has relatively high electricity demand, while every valley has relatively low electricity demand. 
Every $peak\in \mathbb{T}$ and every $valley\in \mathbb{T}$ consists of multiple $hhp$s.
Furthermore 
$\forall hhp_i, hhp_j\in \{Valleys\cup Peaks\}:hhp_j\neq hhp_i$, we have
$  hhp_i\cap hhp_j= \emptyset$,
and  $\bigcup_{hhp\in  \{Valleys\cup Peaks\}}hhp=\mathbb{T}.$

\begin{example}\label{exampleWeelemd}
The following example demonstrates electricity prices for the day-ahead Spot market of the weekend of the 4th August 2023. 
\begin{figure}[ht]
\centering
   \includegraphics[scale=0.6,trim={0 13 0 0}]{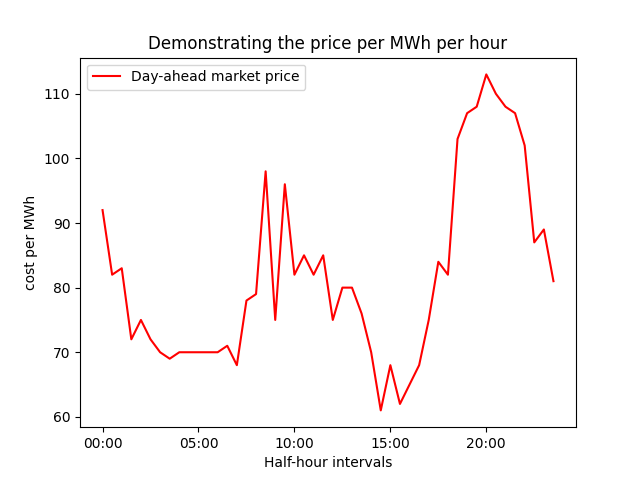}
    \caption{Day ahead Prices at the weekend of 4th August 2023 
    }   \label{fig:exampleWeekend}
\end{figure}
\end{example}
Figure~\ref{fig:exampleWeekend} depicts the day-ahead Spot market (data from epexsport.com).


Consider the  division of $\mathbb{T}$ into peaks and valleys  as follows: 
$$Peaks=\{ 00:00-00:30, 8:00-12:00, 17:00-24:00\}$$ 
which describe high electricity demand periods, and $$Valleys=\{ 01:00-07:30, 12:30-16:30\}$$ which describe off-peak and low electricity demand periods.

Let us denote $\forall hhp\in \mathbb{T}:D^{hhp}$ as the expected   demand\footnote{See \url{https://www.iaea.org/topics/energy-projections}.} for electricity at  $hhp$  
where
$\forall hhp\in \mathbb{T}: D^{hhp}\in \mathbb{R}_{\geq 0}$. Denote $D=(D^{hhp_1},D^{hhp_2}, ..., D^{|\mathbb{T}|} ) $ as the vector of expected demand, 
 and add a \textit{safety margin} $(sm\geq 0)$ to the vector of expected demand
 ${D}$, where $\forall hhp\in \mathbb{T}:\hat{D}^{hhp}=D^{hhp}+sm$,\footnote{This term  includes a safety margin, such that the true expected demand for a given $hhp$ is lower than $\hat{D}^{hhp}$. We leave the analysis of determining the optimal safety margin $ms$ open for future research (see proposition~\ref{propBoundForProbabilty}).} such that
\begin{equation}\label{eqDemand}
\forall hhp\in \mathbb{T}: \hat{D}^{hhp} = D^{hhp}+ \ sm \ \ \ and \ \ \hat{D}=(\hat{D}^{hhp_1}, \hat{D}^{hhp_2}, ..., \hat{D}^{hhp_{|\mathbb{T}|}}).
\end{equation}

Denote the Spot market prices:  $m_{day-ahead}^{hhp}, m_{intra-day}^{hhp}, m_{balancing}^{hhp}$: such that $m_i^{hhp} $ describes the  Spot market price at $ hhp\in \mathbb{T} $ for the $i^{th}$ Spot market, where  $i\in\{day-ahead,intra-day,balancing\}$, such that
\begin{equation}\label{eqMarketPrics}
\forall hhp\in \mathbb{T}: {m}^{hhp}_{i} \in \mathbb{R} \ \ \ where \ \ M_{i}=(m^{hhp_1}_{(\cdot)}, m^{hhp_2}_{i}, ..., m^{hhp_{|\mathbb{T}|}}_{i}).
\end{equation}

We will refer to Spot market prices as  $\forall hhp\in \mathbb{T}:m^{hhp}$  and  $M$ when clear from the context regarding the specific Spot market $i\in\{day-ahead,intra-day,balancing\}$.

We consider a set of fleets, where each fleet has a number of Electric Vehicles (EV)'s, and  a certain quantity of  $kW$  which the fleet can utilize or export.
Each fleet has private information concerning its scheduling costs, energy costs, probability of default for  its  contracts and
\textit{Battery Usage Deterioration Cost ($c^{BD}$)}, where $c^{BD}$ describes the expected cost of battery deterioration  for charging and discharging  1 $kW$\footnote{
$c^{BD}$ is dependent on many attributes such as age of battery, speed of discharge, etc. (see~\cite{fernandez2013capacity})}.

\section{The Hour-Scheduling Mechanism Setting}\label{sectionHourSchedulingMec}

The Hour-Scheduling mechanism  can be implemented   by fleets wishing  to export energy from their bi-directional batteries. 
As stated above, this mechanism allows a set of fleets to collectively store energy during off-peak hours for export during peak hours,  thus reducing carbon emissions and energy costs.
The Hour-Scheduling mechanism promotes  fleets' collective cooperation  of energy storage  by inducing fleets  to trade energy in accordance with  respective energy demand throughout the day. 
By enabling the $N$ fleets to engage in contracts  stating  quantities of $kWh$ they wish to export at  specific $hhps$ to the platform,  we enable the system to divert importing $kWh$ from the grid at peak hours to off-peak hours.

Following~\cite{Clarke71}, we present the 
 Hour-Scheduling mechanism implemented for the day-ahead Spot market, which  can also be applied to multiple intra-day markets separately, or in succession.

We assume the following:

Consider $N$ fleets. Each fleet $n\in N$ owns a number of EV.

Fleets import $kW$ from the grid via their EVs' bi-directional batteries.
Fleets' imported $kW$ consist of \textit{bundles} of $kW$ we denote as $\omega\in WB$, where each bundle $\omega$ consists of a quantity of $kW$,  and $WB$  is the set of  bundles of $kW$ (i.e. the set of $\omega$ EVs have to export), and each  bundle $\omega\in WB$ may be exported only as a whole. When an EV exports their bundle $\omega$ of $kW$ to the grid during an $hhp\in \mathbb{T}$ they are in essence exporting their bundle $\omega$ of $kWh$.

We assume that fleets import $kWh$  during off-peak hours, and offer  to export available bundles $\omega\in WB$ during peak hours.

In the Hour-Scheduling mechanism each fleet  $n \in N$, offers contracts to the system of the form   $j=(n, b, \tilde{hhp},\omega,f_l)\in \mathbb{J}$, where  $n$ is the fleet  offering  contract $j$, $b$ is the bid fleet $n$ announces for contract $j$ which describes the minimal payment  fleet $n$ will accept for each exported $kWh$,  $\tilde{hhp}$  states the $\tilde{hhp}\in \mathbb{T}$ in which the energy transfer is to take place (for example 15:00-15:30), $\omega\in WB$ is the specific bundle of $kW$ offered for export in contract $j$, and $f_{\ell}$ is the penalty  paid by fleet $n$ if contract $j$ is not honored.  
%
%
$\ell:\mathbb{J}\rightarrow \mathbb{R}_{\geq 0}$ describes the amount of $kWh$ offered in contract $j\in \mathbb{J}$, and  $\mathbb{J}$ describes the set of all possible contracts that can be offered.   
 %

%
Each  fleet $ n\in N$ may offer multiple contracts $J^{offered}_n\subseteq J^{offered}\subseteq \mathbb{J}$, where ${J}^{offered}$ is the   set of  contracts that all fleets $n\in N$  offer,  and $J^{offered}_n$ is the set of contracts offered by fleet $n\in N$.
That is $J^{offered}=\{j\in J^{offered}_n|n\in N \}\subseteq \mathbb{J}$.

For any $j\in J^{offered}$, let $n:\mathbb{J}\rightarrow N$ indicate the fleet that offered contract $j\in \mathbb{J}$.
Let $\omega:\mathbb{J}\rightarrow WB$ indicate the specific $kWh$ bundle offered  in contract 
 $j\in \mathbb{J}$, where the bundle $\omega$ is the excess energy in the battery of an EV belonging to fleet $n$. Let $\ell:WB\rightarrow \mathbb{R}_{\geq 0}$ indicate the quantity of $kWh$ offered in bundle $\omega\in WB$.
 Let $b:\mathbb{J}\rightarrow \mathbb{R}_{\geq 0}$ indicate the  bid in pence per $kWh$ 
 that fleet $n(j)$ offers per  $kWh$ in contract $j\in \mathbb{J}$, and let $hhp:\mathbb{J}\rightarrow \mathbb{T}$ indicate the $hhp\in \mathbb{T}$ at which contract $j\in \mathbb{J}$ should be carried out.  Each EV in fleet $n
\in N$ can offer to export their bundle  $\omega$ only once, and for simplicity, we assume that they  can only export this energy as a single bundle $\omega$.

Consider the list of distinct offered bundles of  $kW$  for export by fleets' $WB=\{\omega'| \exists j\in J^{offered}:\omega(j)=\omega'\}$ where 
$\forall \tilde{\omega}\in WB:|\sum_{\omega\in WB}(\omega'=\tilde{\omega})|\leq 1$. WB is arbitrarily ordered.

Let us denote $\forall \omega'\in WB: J^{\omega'}=\{  j\in J^{offered}: \omega(j)=\omega'\}$, i.e. the set $J^{\omega'}$ describes all contracts that offer to export the specific $\ell(\omega')$ $kW$ in bundle $\omega'\in WB$.
Note that  each $\omega\in WB$ is distinct and it's corresponding $\ell(\omega)$ $kW$ can not be exported more than once, although  a specific $\omega\in WB$ may be offered   multiple times for different $hhp\in \mathbb{T}$ in different contracts during the same $peak$;  a fleet may wish to export bundle $\omega$ during a specific $peak\in Peaks$ but may be flexible regarding the particular $hhp\in peak$\footnote{We  assume  $\forall\omega\in WB:\ell(\omega)$  were imported during the preceding valley period}.
Therefore $\forall \omega\in WB:$ there may exist multiple contracts in $J^{\omega}=\{j,..., j'\}\subseteq J^{offered}$ such that $\omega(j)=...=\omega(j')=\omega$, where 
$\forall j\neq j'\in J^{\omega}$ we have $ hhp(j) \neq  hhp(j')$ and $\{hhp(j),   hhp(j')\}\in peak$ indicate that fleet $n(j)$ is offering to export the \textit{same} bundle $\omega$ at different $hhp$s$\in peak$,
although, only a single contract offering $\omega\in J^{\omega}$ may be accepted.

Note that $J^{offered}=\bigcup_{\omega\in WB}J^{\omega}$ and $\forall \omega'\neq \omega''\in WB$ we have $J^{\omega'}\cap J^{\omega''}=\emptyset$.

Let us denote  
the $k^{th}$ bundle in the vector $WB$ as $\omega_k\in 
 WB$ and/or as $WB_k$, where $ k\in [1, |WB|]$.

Let us denote the vector of sets of offered contracts 
as $\overline{WB}(J^{offered})$ 
such that $\forall k\in  [1,|WB|]:  \overline{WB}_k(J^{offered})=\{ j\in J^{offered}|\omega(j)=WB_k \}=J^{\omega_k}$.
Specifically, $\overline{WB}(J^{offered})=(J^{WB_1}, ..., J^{WB_{|WB|}}  )$. 

 We abbreviate $\overline{WB}(J^{offered})$ to $\overline{WB}$ when clear from the context. 

 The platform perceives a vector of sets of contracts $\overline{WB}$  derived from offered contracts $J^{offered}$   as well as the electricity demand $\forall hhp\in \mathbb{T}:\hat{D}^{hhp}$ and a Spot market prices $\forall hhp\in \mathbb{T}:m^{hhp}\in M$. Subsequently, the platform  accepts a subset of these contracts $\hat{J}(\overline{WB},\hat{D},M)\subseteq J^{offered}$ to be implemented by their offering fleets, aiming to maximize \textit{society's savings}. 
Let $SocSave(J,\hat{D}, M)$  describe society's savings given contracts $J\subseteq\mathbb{J}$, demand $\hat{D}\in \mathbb{R}_{\geq 0}^{|\mathbb{T}|}$ and market prices $M\in \mathbb{R}^{|\mathbb{T}|}$ (see appendix~\ref{appedixCostRed}).
Given $\overline{WB}$, $M$ and $\hat{D}$, let $J^*(\overline{WB},\hat{D}, M)\subseteq J^{offered}$ describe an optimal set of contracts that maximizes social welfare, i.e. minimizes society's savings $SocSave(\cdot)$ (see algorithm~\ref{algValue}).

Finding an optimal set of contracts $J^*(\overline{WB},\hat{D},M)\subseteq \mathbb{J}$ is an NP-hard problem, as it  can be reduced from  the knapsack problem, as shown in proposition~\ref{prop100624}.

\begin{proposition}\label{prop100624}
For any set of fleets' offered contracts $J^{offered}$, for any market prices $M\in \mathbb{R}^{|\mathbb{T|}}$ and for any electricity demand $\forall D\in \mathbb{R}_{\geq 0}^{|\mathbb{T|}}$, finding an optimal set of accepted  contracts $J^{*}\subseteq J^{offered}$ 
that maximizes social welfare, i.e. 
$$\forall \tilde{J}\subseteq J^{offered}:SocSave(J^{*}, D,M)\geq SocSave(\tilde{J},D,M)$$
 is NP-hard.
\end{proposition}

\begin{proof}

Proof is by reduction from the Knapsack problem~\cite{karp1972complexity}, where we consider a Knapsack instance $\{(v_i,x_i)\}_{i\in \tilde{N}}$ (volume, worth). 
  
We consider a single  $\overline{hhp}\in \mathbb{T}$, a single fleet $\overline{n}$,  and a fixed fine for default $\overline{f_{\ell}}=0$. 
We consider a knapsack instance  $\{(v_i,x_i)\}_{i\in \tilde{N}}$, where each $i\in \tilde{N}$ is equivalent to a bundle $\omega\in WB$  thus:   $\forall i\in \tilde{N}:(v_i,x_i)$, $\exists \omega\in WB:$ such that there exists a contract $\exists j\in J^{offered}$, where, $b(j) = \frac{x_i}{v_i}$, $m^{\overline{hhp}}=max_{i\in \tilde{N}}[\frac{x_i}{v_i}]\cdot 10$ and $\ell(\omega)=v_i$. 
The demanded quantity  $D^{\overline{hhp}}$ shall be equivalent to the size of sack $ Z$.

Therefore, $\exists J^*\subseteq J^{offered}:\sum_{j\in J^{*}}\ell(\omega(j))\leq D^{\overline{hhp}}$ and $\sum_{j\in J^{*}}b(j)\cdot \ell(j)=Q$ if and only if there exists a set of items of total worth $Q$ that fits in sack $Z$.
\end{proof}

Note that finding an optimal set of contracts $J^*$  is also a variation  of the  bin packing problem where the   items for packing are the bundles $\omega\in WB$ and the set of bins is constant and equal to $\mathbb{T}$ (see \cite{vazirani1997approximation})\footnote{The volume of each bin is equivalent to it's respective  demand $\forall hhp\in \mathbb{T}$: $D^{hhp}$ and each item in the bin packing problem is equivalent to a bundle $\omega\in WB$, where we wish to fill the bins (the $hhp's)$ with bundles $\omega\in WB$ efficiently.}. The Hour-Scheduling mechanism in section~\ref{secHourSchedulingMechanism} is based on a bin packing setting.

\subsection{The Hour-Scheduling Mechanism Contract}
\label{secHourSchedulingMechanism}

We wish to find the optimal set of accepted contracts $J^{*}\subseteq J^{offered}$  which maximize society's savings, given $\hat{D}$ and $M$, such that
$$\forall \tilde{J}\subseteq J^{offered}:SocSave(J^{*},\hat{D},M)\geq SocSave(\tilde{J},\hat{D},M)$$
where $SocSave(J,\hat{D},M)$ indicate society's  savings from accepted contracts $J$ 
 (see appendix~\ref{appedixCostRed}).

We shall introduce in theorem~\ref{th12032024} a mechanism that returns a set of contracts $\hat{J}(\overline{WB},D,M)$ that is efficient.
Consider the following definition,

\begin{definition}
We say a subset of contracts $J\subseteq J^{offered}$  belongs to the set of feasible contracts $ J^{F}\subseteq J^{offered}$, if for each bundle $\omega\in WB$ there exists at most a single contract $\exists j\in J$ offering to export  $\omega$, i.e. 
$\forall \omega'\in WB:|\{\omega(j)|_{j\in J}= \omega' \}|\leq 1$. We say a subset of contracts $J^*\in J^F$ is  \textit{Efficient and Feasible (EF)} and belongs to the set of contracts $J^{EF}$ if given $ D\in \mathbb{R}_{\geq 0}^{|\mathbb{T}|}$ and given  $  M\in \mathbb{R}^{|\mathbb{T}|}$, then: 
$$  
\forall \tilde{J}\subseteq J^{F}:SocSave(\tilde{J},D,M)\leq SocSave(J^*,D,M).$$
\end{definition}

In order to find an efficient set of accepted contracts $J\subseteq J^{offered}$, we follow \cite{meir2017contract,furini2017effective}, 
and consider dynamic programming over  contracts $J^{offered}$,  and reach an efficient solution in polynomial time. We preform the dynamic programming procedure over the set of  bundles  $\omega\in WB$.

Recall $J^{WB_k}$ denotes the set of contracts  offering bundle $WB_k$, i.e., $J^{WB_k}=\{j\in J^{offered}|\omega(j)=WB_k\}=\overline{WB}_k$.

Let $\forall q\in [1,|WB|]:\breve{{WB}}^q=(\omega_1, \omega_2,..., \omega_{q})\subseteq WB$ denote the vector of bundles in  ${WB}$  proceeding and including  bundle $\omega_q$, and let $\breve{\overline{WB}}^q=(J^{\omega_1},...,J^{\omega_q})
\subseteq \overline{WB}$ denote the vector of sets of contracts offering bundles $\omega\in WB$ preceding and including the set of contracts  $J^{\omega_q}=\overline{WB}_q$. Note that $\Breve{\overline{WB}}^{|WB|}=\overline{WB}$, and $\breve{WB}^{|WB|}=WB$.

Let $\hat{J}(\breve{\overline{WB}}^k,d,M)$   denote a feasible subset of   contracts belonging to the set of contracts $J= \{j\in \bigcup_{q\in [0,k]}J^{WB_q}\}$, given 
 \textit{energy demand $d\in \mathbb{R}_{\geq 0}^{|\mathbb{T}|}:0\leq d \leq D$}\footnote{$0\leq d \leq D$ indicates that for each element $i\in [0,\mathbb{T}]$, $d_i\in [0,D_i]$} and market prices $M$ which is returned by theorem~\ref{th12032024}.

Let $\breve{\overline{WB}}^{q-1}$ describe all sets of contracts $J^{\omega}\in \overline{WB}$ preceding the $q^{th}$ set of contracts in $\overline{WB}$, such that  
$\breve{\overline{WB}}^{q-1}=\breve{\overline{WB}}^q\backslash \overline{WB}_q$.

\begin{theorem}\label{th12032024}
For any market prices $M$,  demand $d'\in [0,D]$, and any set of bundles ${WB}'\subseteq {WB}$,  the accepted contracts $J^{*}(\overline{WB}',d',M)\subseteq J^{offered}$
 (and  the fleets' payments) can
be computed in time polynomial in $|WB|$, $|\mathbb{T}|$ and 
$\Pi_{hhp\in \mathbb{T}}|D^{hhp}|$, where  $\forall \omega\in WB:\ell(\omega)$  consist of bounded integers.
\end{theorem}

\begin{proof}
We provide a dynamic program that computes in poly time  a set of optimal accepted contracts ${J}(\overline{WB},D,M)$ that maximizes the platform's cost reduction (i.e. society's savings - see algorithm~\ref{algValue}) 
$SocSave(J^{*}(\overline{WB},D,M),D,M)$ given demand ${D}$, market prices $M$ and offered contracts $J^{offered}$.

Fix an arbitrary order over $\forall \omega\in WB$. Initialize two tables denoted
$\Ddot{\Ddot{J}}$ and $\Ddot{\Ddot{SocSave}}$, of sizes $(|WB|+1) \times d\times 1$ where $(|WB|+1)$  describes bundles $\omega\in WB$, $d\in (d_{hhp_1}, d_{hhp_2},..., d_{hhp_{|\mathbb{T}|}})$ is a vector of length $|\mathbb{T}|$ which depicts the remaining uncovered energy demand $\forall hhp: 0\leq d_{hhp}\leq D^{hhp}$, and the  vector of prices $M$ is constant and doesn't vary in these tables.  Table $\Ddot{\Ddot{J}}$  describes  potential sets of  contracts $J(\cdot)\subseteq J^{offered}$, and  table $\Ddot{\Ddot{SocSave}}$ describes  society's respective savings $SocSave(\cdot) $. 
Specifically,

\begin{enumerate}
\item The cell $\Ddot{\Ddot{J}}(0, d,M)$ is initialized to $\emptyset$ for $0 \leq  d \leq  (D_{hhp_1}, D_{hhp_2},..., D_{|\mathbb{T}|})$.
\item The cell $\Ddot{\Ddot{SocSave}}(r, d,M)$ is initialized to zero for $d = 0$ and  any 
$r\in [1,..., |WB| ]$.
\item The cell $(r,d,M)$ in  table $\Ddot{\Ddot{J}}$, 
will contain the optimal subset of offered contracts $\hat{J}(\breve{\overline{WB}}^r,d,M)\subseteq \bigcup_{q\in[0,r]}J^{WB_q}$ that maximizes 
$SocSave(\hat{J}(\breve{\overline{WB}}^r, d,M))$, i.e.,   $Max_{J\subseteq\bigcup_{q\in[0,r]}J^{WB_q}}SocSave(J, d,M)$.
\end{enumerate}

For every $\forall r\in [1,...,|WB|  ]$ and every $\forall d: 0\leq d\leq D$:
Given $J(\breve{\overline{WB}}^{r-1}, d,M)$, we  compute $J(\breve{\overline{WB}}^r, d,M)$ by considering the best option to meet $d$
without bundle $WB_r$, i.e., $J(\breve{\overline{WB}}^{r-1}, d,M)$, and by computing the maximal value of adding a contract  $\tilde{j}\in \{J^{\breve{\overline{WB}}_r}\cup \emptyset\}$ to $\hat{J}(\overline{WB}^{r-1}, (d^{hhp(j)}-\ell(\tilde{j}),d^{-hhp(\tilde{j})}),M)$. 

More formally, let
$$j^*=argmax_{j\in J^{\breve{\overline{WB}}^r}}\{SocSave(J(\breve{\overline{WB}}^{r-1}\cup j, (d^{hhp(j)}-\ell(j),d^{-hhp(j)})),M)) \}
$$
 and 
$$SocSave^*=max_{j\in J^{\breve{\overline{WB}}^r}}\{SocSave(J(\breve{\overline{WB}}^{r-1}\cup j, (d^{hhp(j)}-\ell(j),d^{-hhp(j)}),M))\}
$$
\normalsize
If $SocSave^*
\leq  SocSave(\hat{J}(\overline{WB}^{r-1},d,M))$  then set $\Ddot{\Ddot{J}}(\breve{\overline{WB}}^r, d,M) = \Ddot{\Ddot{J}}(\breve{\overline{WB}}^{r-1} , d,M) $ and
\small
\newline
    $\Ddot{\Ddot{SocSave}}(J(\breve{\overline{WB}}^r,d,M) )=\Ddot{\Ddot{SocSave}}(\hat{J}(\overline{WB}^{r-1}, d,M))$. 
Otherwise, set 
$J(\breve{\overline{WB}}^{r}, d,M) =
J(\breve{\overline{WB}}^{r-1},(d^{hhp(j^*)}-\ell(j^*),d^{-hhp(j^*)}),M) \cup { j^*
}$ 
\normalsize
and 
\small
$SocSave(J(\breve{\overline{WB}}^r,d,M))=SocSave(J(\breve{\overline{WB}}^{r-1},(d^{hhp(j^*)}-\ell(j^*),d^{-hhp(j^*)}),M) \cup { j^*
}, D, M)$.
\normalsize

We can  solve the optimization problem to find $\hat{J}(\overline{WB},d,M)$ by using binary search on the value of  $SocSave(\hat{J}(\breve{\overline{WB}}^r,d',M))$ for $ \forall d' \leq  d$:$J(\breve{\overline{WB}}^r, d',M)$, and select the solution with the maximal $SocSave (J(\breve{\overline{WB}^r},d',M))$ among them. 

Due to induction, $\hat{J}(\overline{WB},D,M)$ is an optimal set of contracts, i.e.
$$\forall \tilde{J}\subseteq J^F:Soc(\hat{J},D,M)\geq Soc(\tilde{J},D,M).$$
Therefore, we term the set of contracts $\hat{J}(\overline{WB},D,M)$ as $J^*(\overline{WB},D,M)$ or if clear from the context, as $J^*$.
\end{proof}

After contracts $\hat{J}(\overline{WB},D,M)$ have been chosen, each fleet will receive a payment for their accepted contracts  based on the VCG-rule~\cite{Clarke71} which is calculated as follows: 
\begin{equation}\label{eqPaymentRule} 
\forall n\in N: payment^n(J^{offered},D,M)
= SocSave_{-n}( \hat{J}(\overline{WB},D,M))- SocSave(\hat{J}(\overline{WB}_{-n},D,M))
\end{equation}
where $SocSave_{-n}(\hat{J}(\overline{WB},D,M))$ is equal to the expected savings of the platform  due to contracts $\hat{J}(\overline{WB},D,M)$ excluding the cost of fleet $n\in N$ (see proposition~\ref{propTruthfulIRBB}(\ref{proptruth})), that is
$$SocSave_{-n}(\hat{J}(\overline{WB},D,M),D,M)=
SocSave(\hat{J}^*(\overline{WB},D,M),D,M)+\sum_{j'\in \hat{J}_n(\overline{WB},D,M)}b(j')\ell(j'),
$$ 
\normalsize
and where $\overline{WB}_{-n}$ are the contracts  in $\overline{WB}$ excluding the contracts  offered by fleet $n$, that is $\overline{WB}_{-{n}}=\{J^{\omega}\in \overline{WB}|\nexists j'\in J^{\omega}:n(j')=n  \}$.

Note, the payment each fleet $n\in N$ receives  is equal to the  additional value that fleet $n$ brings society. 
The payment is paid to the fleets upfront after accepting contracts ${J}^*(\overline{WB},D,M)$, regardless of the fine $f_{\ell}$ fleets may pay in practice if they fail to honor their  of contracts.

\begin{remark}\label{rem150824}
Note that  dynamic programming for this problem is not considered  efficient   with $\mathbb{T}>10$~\cite{jansen2013bin}. This issue can be resolved by considering theorem~\ref{th12032024} for each peak separately rather than for an entire day with $48$ $hhp$'s, and to have $\mathbb{T}$ equal only to the set of $hhp$'s in a single peak with no more  than ten $hhp$s. 
\end{remark}

\begin{remark}\label{rem140724}
The contract for independent EVs is similar to that of fleets, where we assume that each fleet has a single EV. Every fleet (with a single EV) may offer a single $\omega\in WB$ per peak by offering multiple contracts for different $hhp$'s for each peak period, where only a single contract per peak period can be accepted.
\end{remark}

\begin{remark}\label{rem310125}If an EV whishes to export its' excess energy in two bundles, then that EV could be considered as a fleet with two virtual EVs. However, the fleets' virtual EVs' schedules may be dependent on each other. We leave this open for future research. 
\end{remark}

Consider the following definition,
\begin{definition}\label{defMimport}
$m^{imported} $  indicates  the   pence 
  per $  kWh$  of  the $ \ell(j) $  which fleet $n(j)$ originally imported from the energy market.
  \end{definition}

Let $p:\mathbb{J}\rightarrow [0,1]$ denote  the private probability of fleet $n(j)$ fulfilling contract $j\in \mathbb{J}$. 
Let $c:\mathbb{J}\rightarrow \mathbb{R}_{\geq 0} $ denote the private expected   cost of fleet $n(j)$ attempting contract $j\in \mathbb{J}$.
For example, the expected cost of  contract $j$ for fleet $n(j)$ per may be calculated thus\footnote{We assume the  fleets' costs $\forall j\in \mathbb{J}:c(j)$ are endogenously determined (for further  discussion see~\cite{meir2017contract}).}:
\begin{equation}\label{eqCost}
\mathbb{E}[c(j)]=p(j)\cdot[ m^{imported}\cdot \ell(j)  +c^{BD}\cdot(\ell(j))]\cdot \ell(j)+f_{\ell}\cdot (1-p(j))+sceduling\_cost(j)  \end{equation}
\normalsize

We say that the \textit{utility}, denoted $u^n(\cdot)$, that fleet $n\in N$  attains is a function of $\overline{WB}$,  $\hat{D}$ and $M$. Specifically, $u^n(\overline{WB},\hat{D},M)=payment^n(\overline{WB},\hat{D}, M)-\sum_{j\in \hat{J}_n(\overline{WB},\hat{D},M)}\mathbb{E}(c(j))$ 
(see equations (\ref{eqCost}), (\ref{eqPaymentRule})),
and the  utility  the system operator will attain  shall be 
$$u^{system}(\overline{WB},\hat{D},M)
=\sum_{hhp'\in \mathbb{T}}[Min\{ \hat{D}^{hhp'} ,\sum_{j\in \hat{J}(\overline{WB},\hat{D}):hhp(j)=hhp'}\ell(j)\}]\cdot m^{hhp'}-\sum_{n\in N}payment^n(\overline{WB},\hat{D},M).$$
\normalsize

We assume that fleet owners and the system aim to maximize their own expected utility.

\section{Analysis of the Hour-Scheduling Mechanism}\label{secAnalysis}

\begin{theorem}\label{thEfficiency}
Given $WB'\subseteq WB$, $d\in \mathbb{R}^{|\mathbb{T}|}:0\leq d\leq D$ and $M$, the Hour-Scheduling mechanism leads to an efficient and feasible set of accepted contracts in $J^*\in J^{EF}$.
\end{theorem}

\begin{proof}
The proof follows by induction through the Hour-Scheduling mechanism (see theorem
~\ref{th12032024}).
\end{proof}

\begin{definition}\label{defTruthful}
A mechanism is \textit{truthful} if it is the dominant strategy for any fleet offering  a contract $j\in J^{offered}$ to announce a bid $b(j)$ equal to her expected cost according to equation~(\ref{eqCost}). 
\end{definition}

\begin{definition}\label{defBB}
A mechanism is \textit{Budget Balanced  (BB)} if the central platform 
receives a non-negative expected utility by applying the mechanism. 
\end{definition}

\begin{definition}\label{defIR}
A mechanism is \textit{Individually Rational (IR)} if 
by
participating in the 
mechanism, each fleet attains non-negative utility in expectation.\end{definition}

Following \cite{nisan2007computationally,jansen2013bin}, we show the following proposition~\ref{propTruthfulIRBB}.

\begin{proposition}\label{propTruthfulIRBB}
The Hour-Scheduling mechanism  described in section~\ref{secHourSchedulingMechanism} satisfies \begin{enumerate}
\item \label{proptruth} Truthfulness.
\item \label{propIR} Individual Rationality.
\item \label{propBB} Budget Balanced.\footnote{Note that definition~\ref{defBB} need not consider cases of contracts' default as in the case of default the fleet's fine should cover the resulting loss.}
\end{enumerate}
\end{proposition}

\begin{proof}
See appendix~\ref{apProofs}.
\end{proof}

\begin{lemma}\label{lmProbEstimation}
Fleet $n\in N$ will not offer  a contract $j=(n,b, hhp,\omega,f_{\ell})$ unless $p(j)\geq \frac{f_{\ell}(j)+c(j)-b(k)}{f_{\ell}(j)}$, such that ${p}(j)$ is fleet $n$'s private information regarding  the probability   of fulfilling contract  $j$,  the cost of attempting contract $j$ is    
$c(j)=p(j)\cdot [c^{BD}\cdot\ell(j)+m^{hhp^{imported}(j)}\cdot \ell(j)]+f_{\ell}\cdot \ell(j)\cdot (1-p(j))+Sceduling \ cost $, and $m^{hhp^{imported}}(j)$ is according to definition~\ref{defMimport}.    
\end{lemma}

\begin{proof}
Note that fleet $n$ will not offer  contract $j$ unless its expected utility is non-negative. Because the bid $b(j)$ is the minimal payment fleet $n$ will attain for purchasing contract $j$ (see proposition~
\ref{propTruthfulIRBB}(\ref{propIR})), fleet $n$ will not offer to buy  contract $j$ according to equation~(\ref{eqCost}), unless: 
$ b(j)\geq c(j)= p(j)\cdot [c^{BD}\cdot \ell(j)+m^{hhp^{imported}(j)}\cdot \ell(j)]+f_{\ell}\cdot \ell(j)\cdot (1-p(j))+Sceduling \ cost $. 
Therefore,  we have:
\begin{equation}\label{eq10032023}
p(j)\geq \frac{f_{\ell}(j)+c(j)-b(j)}{f_{\ell}(j)}
\end{equation}
\end{proof}
\begin{lemma}\label{lm141123}

A lower bound for the probability of fulfilling an offered contract $j\in \mathbb{J}$ is  
\begin{equation}\label{eqLowerBoundProbabitly}
\hat{p} (j)=(f_l(j)- b^n (j))/f_l(j).
\end{equation}
\end{lemma}

\begin{proof}

Due to lemma~\ref{lmProbEstimation} and proposition~\ref{propTruthfulIRBB}, the Hour-Scheduling mechanism induces each fleet to announce 
a bid $b(j)$ (per contract) such that
$$-b(j)+f_{\ell}\cdot (1-p(j))=p(m^{hhp^{imported}(j)}\ell(j)+c^{BD}\cdot(\ell(j)))+(1-p(j))\cdot f_{\ell} + scheduling(j)=c(j)$$
where $p(j)$ is the probability of fulfilling contract $j$,  $m^{hhp^{imported}}(j)$ was the cost of $1$ $kW$  for fleet $n$  (when fleet $n$ originally imported their $kW$) and  $c^{BD}\cdot(\ell(j))$ is the battery deterioration cost for contact $j$. $p(j)$, $m^{hhp^{imported}}(j)$ and $c^{BD}(\ell(j))$  are private information.

In order to assess a conservative estimation of the probability of success of a contract $j$, we assume the cost $c^n(j)$ 
is equal to zero (i.e. $\hat{c}^n(j)=0$, where $\hat{c}^n(j)$ is a lower bound on the cost), and attain the 
following:
$$\cfrac{c^n(j)-b(j)+f_{\ell}}{f_{\ell}}=p(j)\geq \cfrac{f_{\ell}-b(j)}{f_{\ell}}=\hat{p}(j)$$

where $p(j)$ is the true probability of satisfying the contract, and $\hat{p}(j)$ is it's lower 
bound.
\end{proof}

Note that  the assessed probability for  contracts' default    (lemma~\ref{lm141123}) may be higher than their true probability of default. 
These inaccurate assessments of the probability  do not affect the fleets' payments or their allocated contracts in the Hour-Scheduling mechanism, but may affect fleets' payments in the Frequency Regulating Mechanism (see section~\ref{secFreq}).

\section{The Frequency Regulation Mechanism }\label{secFreq}

Frequency control refers to the need to ensure that the grid maintains a reliable power supply,
so that the power grid does not collapse. This  control is maintained through a number of avenues including ancillary services   provided to support the continuous flow of electricity, ensuring the demand for electrical energy is met in real time
with a stable voltage range (for example,
 regulators in the US require  voltage range to be  ±5\% of the nominal voltage\footnote{See \url{https://www.next-kraftwerke.com/knowledge/ancillary-services}, as well as~\cite{luo2018review}.}).

Although contracts were planned and applied to cover  energy demand,  situations often arise in which for certain $hhp\in \mathbb{T}$  energy demand  is not covered by existing energy supplying contracts  due to   fleets'  defaulting from their agreed contracts or due to actual energy demand exceeding  expected  energy demand. In such cases, the remaining demand is  not met and is accommodated either through voltage lowering,  electricity blackouts, or is covered through the {Balancing market} which can increase carbon emissions and be quite costly. 
Furthermore, a situation may arise in which there is an excess supply of energy at a certain $hhp\in \mathbb{T}$, which can overload the grid~\cite{denholm2019timescales}.

We introduce the \textit{Frequency Regulating (FR)}  mechanism for exporting and importing energy alongside the Balancing market which 
enables EV owners who have excess $kW$ (or have not fully charged batteries) to export (or import)  energy for a profitable fee, and thus reduce their own costs as well as  society's costs  in addition to carbon emissions.

The FR Mechanism enables EVs to connect to the platform for discharging or charging energy. 
An EV may connect at any $hhp\in \mathbb{T}$ to the platform, and is paid for maintaining  connection with the platform with availability for energy export or import, as well as for energy exported from (imported to) their battery. 

In order to trade energy efficiently, EVs must consider their batteries' limitations and characterizations. Specifically, an Electric Vehicle's battery has a $SoC$ \textit{(State of Charge)} which describes the number of $kW$ their battery is charged. In order to lessen  battery degradation, on average and depending on the specific features of the battery, the SoC should remain between   20\% to 80\% of the total maximal amount of $kW$ that can be charged~\cite{fernandez2013capacity, thomas2009fuel}.

We consider a set of $M$ EVs that are interested in importing (exporting) $kWh$ to (from) the grid.
We denote by $plugged(hhp)\subseteq M$ the set of EVs connected  to the platform at $hhp\in \mathbb{T}$.
Let
 $x^m_{max}(hhp)$ ($x^m_{min}(hhp)$): describe the desirable maximal (minimal) level of $SoC$ in the battery of EV $m\in M$ at the end of $hhp\in \mathbb{T}$. Specifically,  
$\forall m\in M$, $\forall hhp\in \mathbb{T}:SoC^m(hhp)$ is the $SoC$ in EV $m$'s  battery at $hhp\in \mathbb{T}$. 

Battery discharging has a negative impact on EVs' battery life~\cite{guo2017bidding}, i.e. batteries suffer extra life loss for discharging due to battery degradation~\footnote{The battery life loss is assumed to be affected by a number of aspects, including the battery purchase cost, battery capacity,  designed life cycles, etc.}.
For simplicity, we assume in this paper that the cost of battery degradation $c^{BD}$ is constant for every exported $kWh$ and  is common knowledge.

The FR mechanism is applied as follows:

\begin{enumerate}
\item $\forall hhp\in \mathbb{T}$: each EV $m\in plugged(hhp)$  announces the minimal (maximal) level of $SoC$ they are willing to have in  their battery at the end of  $hhp$, i.e. $x_{min}^m (hhp)$ ($x_{max}^m (hhp)$). The platform reads the $SoC^m(hhp)$  of every $m\in plugged(hhp)$.

\item The platform calculates the available quantity of energy for export (import), i.e. $x^m_{availableEX}(hhp)$ ($x^m_{availableIM}(hhp)$), where: 
 $x^m_{availableEX}(hhp)=SoC^m(hhp)  - x_{min}^m (hhp)$, and $x^m_{availableIM}(hhp)=x_{max}^m (hhp)-SoC^m(hhp)$.
 
\item The platform imports  (exports) $kWh$ from (to) all EVs $m\in Plugged(hhp)$ that are connected to the platform via the FR mechanism according to energy demand and their availability to export (import) $kWh$.
\item Any remaining $kWh$ which are still not covered, are supplemented by importing 
$kWh$ from the Balancing market at $m^{hhp}_{Balancing}$
per $kWh$, and any remaining excess $kWh$ may be exported to the Balancing market or curtailed\footnote{See \url{https://www.next-kraftwerke.com/knowledge/curtailment-electricity}.}.

\end{enumerate}

Let us  denote  $J^{active}$
 as the set of accepted contracts 
through the Hour-Scheduling mechanism (see section~\ref{secHourSchedulingMechanism}) and via other  
exporting and importing energy mechanisms. i.e.  
$J^{active}$ describes the set of active contracts that have 
been accepted and haven't defaulted. 
For simplicity, in this paper we will analyze the FR mechanism  under the assumption that  contracts $J^{active}$ consist only of the accepted contracts $\hat{J}(\overline{WB},\hat{D}, M)$ from the Hour-Scheduling mechanism (see section~\ref{secHourSchedulingMechanism}), i.e. $J^{active}\leftarrow J^{accepted}$. 

Let  $p_{hhp}(y,J^{active})$ denote the true probability of $y$ $kWh$ discharged to the grid via contracts $J^{active}$ at $hhp\in \mathbb{T}$, and let $\hat{p}_{hhp}(y,J^{active}$) denote an assessed   probability that  $y$ $kWh$ will   be supplied given contracts 
$J^{active}$ at $hhp\in \mathbb{T}$
  (see proposition~\ref{prop20012023b}).

Recall $D^{hhp}$ is the actual (see section~\ref{section2}) expected demand at $hhp$, and let $X_{hhp}^{max}(J^{active})$ denote the maximal exported energy if all contracts $J^{active}$ are honored, such that: 
\begin{equation}\label{eqXmax}
    \forall hhp'\in \mathbb{T}:X_{hhp'}^{max}(J^{active})=\sum_{j\in J^{active}: hhp(j)=hhp'}\ell(j)
\end{equation}

For every   $y\in [0,D^{hhp}]|_{hhp\in \mathbb{T}}$, there is a probability of  $p_{hhp}(y,J^{active})$ that $(D^{hhp}-y )$ $kWh$ will not be supplied, and for cases where $D^{hhp}\leq X_{hhp}^{max}$, then
for every $y\in [D^{hhp}, X_{hhp}^{max}]|_{hhp\in \mathbb{T}}$,  
there is a probability of  $p_{hhp}(y, J^{active}))$ of an excess of $(y-D^{hhp})$ $kWh$.

The FR mechanism offers the following contract to each EV $m\in plugged(hhp)$ that is connected to the platform at $hhp\in peak\in \mathbb{T}$,

\small
\begin{equation}\label{eqPaymentBalance}
payment^m_{FR}(x^{max}_m(hhp), x^{min}_m(hhp),hhp , SoC^m(hhp), J^{active}, plugged({hhp}))=
\end{equation}
$$=\sum_{y\in [0,...,D^{hhp}]}const^{EX} \cdot \hat{p}_{hhp}(y,J^{active})\cdot m^{hhp}_{Balancing}\cdot y \cdot \cfrac{x^m_{avilableEX}(hhp) +x^m_{exported}(peak^{-hhp})}{\sum_{q\in plugged(hhp)}(x^q_{avilableEX}(hhp) +x^q_{exported}(peak^{-hhp}))}+$$
$$+\sum_{y\in [D,...,X^{max}_{hhp}]}const^{IM} \cdot \hat{p}_{hhp}(y,J^{active})\cdot m^{hhp}_{Balancing}\cdot y \cdot \cfrac{x^m_{availableIM}(hhp)+x^m_{imported}(peak^{-hhp})}{\sum_{q\in plugged(hhp)}(x^q_{availableIM}(hhp)+x^q_{imported}(peak^{-hhp}))} +$$
$$+X_m^{ exported}(hhp)
\cdot [c^{BD} +m^{hhp}_{day-ahead}]$$
\normalsize

where  $x^m_{exported}peak^{-hhp})$ 
 (or $x^m_{imported}(peak^{-hhp}))$ 
describes the   quantity of $kWh$ that   were exported (or imported) in $peak$ excluding $hhp$ ( i.e.  $\forall hhp'
\in peak\backslash hhp$) by EV  $m\in M$.
This variable   discourages  EV's from endeavoring to  export  (import) their energy at the most profitable $hhp$s and avoid connecting to the platform at less profitable $hhp$s in $peak$, and specifically, it motivates $EV$s to remain connected to the platform, independently of the distribution of profits among all $hhp\in peak$. Furthermore, let
$x^m_{exported}(hhp)$ (or $x^m_{imported}(hhp)$) describe the quantity of $kWh$ that are actually exported (or imported) during $hhp$ by EV $m\in M$.

$ const^{EX}, const^{IM}$  are predetermined constants and are in the platforms' database. $\forall hhp
\in \mathbb{T}: x^m_{imported}(peak^{-hhp}), x^m_{exported}(peak^{-hhp}), x^m_{exported}(hhp), x^m_{imported}(hhp)$ and $ m^{hhp}_{Balancing}$ are  continuously updated in the  platforms' database.  $x^m_{max}$ and $x^m_{min}$ are read by  the platform from any connected EV $m\in plugged(\cdot)$,  based on which $x^m_{availableEX}$ and $x^m_{availableIM}$ are  
 calculated by the platform anew for every connecting EV $m\in M$.

We wish to avoid a situation in which a fleet should manipulate the system by 
avoiding the Hour-Scheduling mechanism in order to attain a  payment from 
the FR mechanism (where there is no penalty for default). 
By adapting `$const^{EX}$', this manipulation can  become less profitable for fleets and 
EV owners.\footnote{Note that while exporting energy in the Hour-Scheduling mechanism via contracts $j\in J^{active}$,  EV $m$ can connect to the Frequency Regulating mechanism, by announcing  $x_{min}^m(hhp)$ that will enable export of $\ell(j)$ and other anticipated $kWh$  obligations as well as the minimal level of $SoC$ the EV will accept. }
Specifically it would be interesting to research the values of $const^{EX}$ and $const^{IM}$  effects on the platforms' profits, society's profits  and on carbon emissions; we leave this open for future research.

The first two expressions in the FR mechanism are constructed in order to motivate EVs to remain connected to the platform for ancillary service requirements, while the third expression `$x_m^{ exported}(hhp)
\cdot [c^{BD} +m^{hhp}_{day-ahead}]$' offers EVs compensation for their exported energy. Specifically, the true costs for exporting energy to the platform are the  EV's battery deterioration cost $x_m^{ exported}(hhp)
\cdot c^{BD} $ as well as the EV's cost for importing the energy  $x_m^{ exported}(hhp)
\cdot m^{hhp}_{imported}$, both of which are private information.

\begin{lemma}\label{lmProbNP}
Computing the probability of $\forall y\in \mathbb{R},\forall hhp\in \mathbb{T}:p_{hhp}(y,J^{active})$  is an NP-hard problem.
\end{lemma}

\begin{proof}
See appendix~\ref{apProofs}.
\end{proof}

Lemma~\ref{lmProbNP} shows that computing $p_{hhp}(y, j^{active})$ is NP-hard.  Therefore, we approximate  this probability 
by considering    contracts'  $j\in J^{active}$  average  probabilities of being honored. We denote this approximated probability as  $\hat{p}(y,J^{active})$.  

Specifically, let us denote the average assessed probability of satisfying a contract per $kWh$ as:
\begin{equation}\label{eq04012025}
\bar{p}=\frac{\sum_{j\in J^{active}} [\ell (j)\cdot \hat{p}(j)]}{\sum_{j\in J^{active}}\ell (j) ,}
\end{equation}

where $\hat{p}(j)$ is determined according to equation~(\ref{eqLowerBoundProbabitly}). 

Let  ${X}_{hhp}^{max}$ describe the  maximal quantity of imported $kWh$   (see equation~(\ref{eqXmax})).
Let us denote $\overline{\ell}$  the average quantity of contracts' offered $kWh$, i.e.

\begin{equation}\label{eq04012025b}
\overline{\ell}=\frac{\sum_{j\in J^{active}}\ell(j)}{|J^{active}|}.
\end{equation}

\begin{proposition}\label{prop20012023b}
Assuming all accepted contracts have the probability $\bar{p}$ of exporting $\bar{\ell}$ $kWh$, according to equations~(\ref{eq04012025}) and~(\ref{eq04012025b}), 
the probability that the mechanism will fail to supply (or will supply an excess of) ${y'}$  $kWh$ to the platform  $\forall hhp\in \mathbb{T}$, is:
\begin{equation}\label{eq01012025}
\hat{p}_{hhp}(y, J^{active})=\cfrac{|\frac{{X}^{max}_{hhp}}{\bar{\ell}}|!}{[\frac{y}{\ell}]!\cdot \left[[\frac{{X}^{max}_{hhp}}{\bar{\ell}}]-[\frac{y}{\bar{\ell}}]\right]!}\cdot [(\bar{p}_{hhp})^{\frac{y}{\bar{\ell}}}]\cdot [(1-\bar{p}_{hhp})^{([\frac{{X}_{hhp}^{max}}{\bar{\ell}}]-[\frac{y}{\bar{\ell}}])}]
\end{equation}
where
$ \ y'={D}^{hhp}-y \ $
  if there is a deficit of $y'$ $kWh$ (and  $ \   y'=y-{D}^{hhp} $ if there is an excess $y'$ $kWh$ supplied).

\end{proposition}

The following proposition shows that the probability $\hat{p}_{hhp}(y, J^{active})$ described in proposition~\ref{prop20012023b} is a lower bound for $p_{hhp}(y, J^{active})$.

\begin{proposition}\label{propBoundForProbabilty}
Given a set of  contracts $J^{active}$,  the approximation $\hat{p}(y, J^{active})$ (see proposition~\ref{prop20012023b}), is a lower bound to the true probability of supplying $y$ $kWh$ at $hhp $ to  the platform, i.e.
$$\hat{p}(y, J^{active})\leq {p}(y, J^{active})$$
\end{proposition}

\begin{proof}
See appendix~\ref{apProofs}.
\end{proof}
 
\begin{corollary}
A risk neutral fleet owner $n \in \mathbb{N}$ will abstain  from obtaining  a contract through the Hour-Scheduling mechanism, in order to obtain  payment through the FR mechanism, if
$$\mathbb{E}[u^n(\hat{J}(\overline{WB},\hat{D},M))
\leq $$
$$\leq \mathbb{E}\sum_{m\in M}[payment_{FR}^m((x^{max}_m(hhp), x^{min}_m(hhp),hhp , SoC^m(hhp), J^{active}, plugged({hhp})))-$$
$$-(x_{exported}^m(hhp)\cdot (c^{BD}+m^{imported})]$$
\normalsize
\end{corollary}
where $m^{imported}$ is in accordance with  definition~\ref{defMimport}, 
which is private information.

Note that a fleet $n$ can add to their bids $ b(j)|_{j\in J_n}$ in the Hour-Scheduling mechanism  the expected loss from not participating in  the FR mechanism.

\section{Conclusion}\label{secConclusion}

V2X technology is a leading direction   for meeting
net-zero objectives which may be applied through electricity trading  in global and local energy markets. Applying V2X technology 
by using  energy trading mechanisms to induce electric vehicles to apply their electric vehicle batteries for  energy storage during off-peak hours for export  during peak hours, enables cost and carbon emission reduction.

We  consider a model enabling the joint effort of EV fleets  participating in energy transaction contracts to help cover   energy demand of the energy system via their batteries' joint energy storage.

We suggest two mechanisms for facilitating energy trading  between electric vehicles and a grid (via V2G) or other entities with electricity demand, such as homes (via V2H), loads (via V2L), and buildings (via V2B):

\begin{itemize}
\item The \textit{Hour-Scheduling mechanism}  for exporting electricity to the platform alongside Spot markets. 
\end{itemize}

This mechanism is   
  a novel application of a VCG mechanism~\cite{Clarke71} in the context of electricity trading,  
allowing fleets to engage in contracts  stating  quantities of $kWh$ they wish to export at  specific $hhps\in \mathbb{T}$ to the platform,  thus enabling the system to divert importing $kWh$ from the grid at peak hours to off-peak hours.

Our approach  simultaneously induces truthfulness while leading to energy trading efficiency (see theory~\ref{thEfficiency}). It considers all $hhp$s of the day and their interrelated-dependencies and not only  single $hhp$s. Furthermore, this mechanism
can be computed in polynomial time (see theorem~\ref{th12032024}). In addition, 
the Hour-Scheduling mechanism may  engage with  fleets  consisting of multiple EVs, each EV with it's own energy consumption schedules and costs, rather than only  independently owned   single EVs.

Furthermore, we introduce the following
mechanism for exporting energy to support the Ancillary Services  market:  
\begin{itemize}
\item The \textit{Frequency Regulating mechanism}  enables frequency regulation via EV's $kWh$ charging and discharging potential. 
\end{itemize}

In order to support the continuous flow of electricity, so that the demand for electrical energy is met in real time,
we suggest the  Frequency Regulating mechanism to alleviate and ease the Balancing market.
This mechanism compensates EV owners for remaining connected  to the platform,  allowing their EV bi-directional batteries to operate as  flexible loads which can import or export energy to and from the platform according to demand when the grid's voltage is not in the required range~\cite{luo2018review}.
By using  EV's batteries for ancillary services,  society's electricity costs as well as society's carbon emissions are reduced.

Future research directions include the analysis of the optimal safety margin $sm$ (see equation~\ref{eqDemand}) in addition to the optimal values $const^{EX}$ and $const^{IM}$ (see equation~\ref{eqPaymentBalance}) in order to maximize expected revenue, social welfare or/and to minimize carbon emissions.  
In addition,  expanding the Hour-Scheduling mechanism to the Forward market as well as to the Ancillary Services market could also be intriguing.

\section{Acknowledgments}\label{secAcknowledgement}

Part of this work was conducted while the authors were employed by Keele University working on the Innovate UK funded project - Zero Carbon Tariffs with Cashback for V2X enabled Non-domestic Customers. We thank colleagues at Qbots Energy for useful discussions.

\bibliographystyle{ACM-Reference-Format}
\bibliography{bibliography}

\begin{appendix}

\section{Society's Savings due to  Contracts $J\subseteq J^{offered}$}\label{appedixCostRed}

The following algorithm~\ref{algValue}   determines  society's   savings (cost reductions) due to applying the Hour-Scheduling mechanism, and returns $SocSave:J\times D\times M\rightarrow \mathbb{R}$ where $M\in \mathbb{R}^{\mathbb{T}}$ describes the vector of market prices,    $D\in \mathbb{R}^{\mathbb{T}}_{\geq 0}$ describes the vector of energy demand, which has  not yet been covered by contracts, and $J\subseteq J^{offered}$ is a set of contracts belonging to $J^{offered}$. In the case where participants are truthful and announce bids equal to their cost for each contract $j\in J:$ $b(j)=c(j)$, then $SoCSave$ is equivalent to social welfare.

For abbreviation we refer at times to $SocSave(\hat{J}(\overline{WB},D,M),D,M)$  as $SocSave(\overline{WB},D,M)$ when clear from the context.

\begin{algorithm}[H]
\caption{Returns  society's cost reduction (saving) ($SocSave(\cdot)$) due to accepting  contracts ${J}\subseteq J^{offered}$, energy demand $D$ and market prices $M$.} 
\label{algValue}
\KwIn{A set of  contracts  ${J}\subseteq J^{offered}$,   
 market prices $\forall hhp\in \mathbb{T}:m^{hhp}$ and  demand $\forall hhp\in \mathbb{T}:D^{hhp}$.  } 
\KwOut{ Society's  savings due to the platform accepting and implementing contracts $J\subseteq J^{offered}$. }
$SocSave=\sum_{j\in {J}}[m^{hhp(j)}\cdot \ell(j)-b(j)] - \sum_{hhp'\in \mathbb{T}}[ m^{hhp'} \cdot Max\{0,\sum_{j'\in J:hhp(j')=hhp'}\ell(j')-D^{hhp'}\}]
$
\Return{ $SocSave$ }
  \end{algorithm}
\section{Proofs}\label{apProofs}

\textbf{Proposition~\ref{propTruthfulIRBB}(\ref{proptruth})}.  The mechanism described in theorem~\ref{th12032024} is truthful.

\begin{proof}
Consider the set of accepted contracts by the Hour-Scheduling mechanism  in theorem~\ref{th12032024}  $\hat{J}(\overline{WB},\hat{D})$. The utility of fleet $n\in N$ due to  the Hour-Scheduling mechanism is: $$u^n(\overline{WB},\hat{D},M)=payment^n(\overline{WB},\hat{D},M)-\sum_{j\in \hat{J}_n(\overline{WB},\hat{D},M)}c(j)$$ 
where $c(j)$ 
is the cost of contract $j$ according to  equation~(\ref{eqCost}). Therefore, due to equation~(\ref{eqPaymentRule}),
\begin{equation}\label{eq170125}u^n(\overline{WB},D,M)=[SocSave_{-n}( \hat{J}(\overline{WB},\hat{D},M))- SocSave(\hat{J}(\overline{WB}_{-n},\hat{D})) ]-\sum_{j\in \hat{J}_n(\overline{WB},\hat{D})}c(j)\end{equation}

Because $SocSave(\hat{J}(\overline{WB}_{-n},\hat{D},M)) $ is independent of fleet $n$'s announcements,  
fleet $n$'s dominant strategy is to announce $\hat{b}_n\in \mathbb{R}^{|J_n|}$ that will  maximize the term:

\begin{equation}\label{eq170125a}
argmax_{\hat{b}_n\in \mathbb{R}^{|J_n|}}\{ SocSave_{-n}( \hat{J}(\overline{\hat{WB}}_n, \overline{WB}_{-n},\hat{D},M))-\sum_{j\in \hat{J}_n(\overline{\hat{WB}}_n, \overline{WB}_{-n},\hat{D},M)}c(j)\}\end{equation}
 where $\overline{WB}_{-n}$ is the vector of sets of contracts induced by the bids offered by fleets $n'\in N\backslash n$,
and where $\overline{\hat{WB}}_n$ is the vector of sets of contracts induced by fleet $n$'s  announced bids $\hat{b}_n\in \mathbb{R}^{|J_n|}$.

For simplicity, we assume that each fleet $n'$s bundles $WB_{n'}$ and contracts $J_{n'}$ are constant although  their bids $\forall j\in J^{n'}:b_{n'}$ may vary.

Note that  algorithm~\ref{algValue} selects the set of contracts $\tilde{J}\subseteq (\overline{\hat{WB}}_n,\overline{WB}_{-n})$  such that

\begin{equation}\label{eq17012025b}\tilde{J}\in argmax_{J\subseteq (\overline{\hat{WB}}_n,\overline{WB}_{-n})}\{SocSave( \hat{J}(\overline{\hat{WB}}_n,\overline{WB}_{-n},\hat{D},M))\}\end{equation}

In addition, note that due to theorem~\ref{thEfficiency}
\small
$$\forall J'\in J^{F}:$$
$$Max_{J\subseteq \overline{WB}}[SocSave_{-n}(\hat{J}(\overline{WB},\hat{D},M),\hat{D},M) - \sum_{j\in J^n}c(j)
  =SocSave(\hat{J}(\overline{WB},\hat{D},M),\hat{D},M)\geq SocSave(J',\hat{D},M)$$
  \normalsize
Therefore, due to equations~(\ref{eq170125}), ~(\ref{eq170125a}) and~(\ref{eq17012025b})
\small
$$Soc(\hat{J}(\overline{WB},\hat{D},M),\hat{D},M)\geq max_{\hat{b}\in \mathbb{R}^{|J_n|}}[\hat{J}((\overline{\hat{WB}}_n,\overline{WB}_{-n}),\hat{D},M)-\sum_{j\in \hat{J}((\overline{\hat{WB}}_n,\overline{WB}_{-n}),\hat{D},M)}c(j)]$$
\normalsize
Thus, we deduce that $u^n(\overline{WB},\hat{D},M)\geq u^n((\overline{WB}_{-n},\overline{\hat{WB}}_n),\hat{D},M)$.

We maintain that the dominant strategy for fleet $n$ is to announce their true value $\forall j\in J^n:\hat{b}(j)=c(j)$, in order to maximize their utility.
\end{proof}

\textbf{Proposition~\ref{propTruthfulIRBB}(\ref{propIR}) }
The Hour-Scheduling mechanism satisfies IR, i.e.
\begin{equation}\label{eqCostLesThanPayment}
\forall n
\in N:\sum_{j\in {J}_n(\overline{WB},\hat{D},M)}c(j)\leq payment^n (\overline{WB},\hat{D},M)
\end{equation}

\begin{proof}

Assume negatively, that $\exists n\in N: \sum_{j\in {J}_n(\overline{WB},\hat{D},M)}c(j)>payment^n (\overline{WB},\hat{D},M)$.
That is, due to equation~(\ref{eqPaymentRule})
$$\sum_{j\in J^{}_n(\overline{WB},\hat{D},M)}c(j)>payment^n(\overline{WB},\hat{D},M)=  SocSave_{-n}( \hat{J}(\overline{WB},D,M))- SocSave(\hat{J}(\overline{WB}_{-n},D,M))  $$

Therefore, 
$$SocSave({J}(\overline{WB}_{-n},\hat{D},M)) > SocSave_{-n}({J}(\overline{WB},\hat{D},M))- \sum_{j\in \hat{J}_n(\overline{WB},\hat{D},M)}c(j) $$

And due to proposition~\ref{propTruthfulIRBB}(\ref{proptruth})
$$SocSave_{-n}({J}(\overline{WB},\hat{D},M))- \sum_{j\in \hat{J}_n(\overline{WB},\hat{D},M)}c(j)\cdot \ell(j) = SocSave(\hat{J}(\overline{WB},\hat{D},M)) $$

That is,
$$  SocSave(\hat{J}(\overline{WB}_{-n},\hat{D},M))>SocSave(\hat{J}(\overline{WB},\hat{D},M))$$
which is a contradiction to proposition~\ref{th12032024}.
\end{proof}
\textbf{Proposition~\ref{propTruthfulIRBB}(\ref{propBB}) }  
    The Hour-Scheduling  mechanism in section~\ref{sectionHourSchedulingMec} satisfies  BB.\footnote{Note that definition~\ref{defBB} does not consider expectation of contracts' default as in the case of default the fleet's fine should cover the resulting loss.}, that is, $$SocSave(\overline{WB},\hat{D},M)\geq \sum_{n\in N}payment^n(\hat{J}(\overline{WB},\hat{D},M),\hat{D},M)$$
\begin{proof}
In this proof, for simplicity we shall assume that 
\begin{equation}\label{eq08012025simplicity}\forall hhp\in \mathbb{T}:\hat{D}^{hhp}\geq \sum_{j\in \hat{J}(\overline{WB},\hat{D},M):hhp(j)=hhp}\ell(j)\end{equation}

Let us prove negatively, and assume:
\begin{equation}\label{eq150724b}
SocSave(\overline{WB},\hat{D},M)<\sum_{n\in N}payment^n(\hat{J}(\overline{WB},\hat{D},M),\hat{D},M)
\end{equation}

Recall from equation~(\ref{eqPaymentRule}) that: 
\begin{equation}\label{eq1507c}
payment^n(\overline{WB},\hat{D},M)
= SocSave_{-n}( \hat{J}(\overline{WB},\hat{D},M))- SocSave(\hat{J}(\overline{WB}_{-n},\hat{D},M))\end{equation}

Therefore, due to equations~(\ref{eq150724b}) and~(\ref{eq1507c}):
\begin{equation}\label{eq091224b}
SocSave( {J}(\overline{WB},\hat{D},M))< \sum_{n\in N}SocSave_{-n}({J}(\overline{WB},\hat{D},M))- SocSave({J}(\overline{WB}_{-n},\hat{D},M))
\end{equation}

Due to algorithm~\ref{algValue}
$$Soc(\overline{WB},\hat{D},M)\leq \sum_{j\in \hat{J}(\overline{WB},\hat{D},M)}(m^{hhp(j)}\cdot \ell(j)-b(j))$$ 

Note that $\forall n\in N$, the additional value fleet $n$ brings  the system is $Soc_{-n}(\hat{J}(\overline{WB},\hat{D},M))-Soc(\hat{J}(\overline{WB}_{-n},\hat{D},M))$ which cannot surpass
the maximal value the system could attain from  fleet $n$'s covering electricity demand through  contracts $J^n$ which is
$\sum_{j\in J_n}(m^{hhp(j)}\cdot 
\ell(j)-b(j))$. Therefore,  

\begin{equation}\label{eq300125}Soc_{-n}(\hat{J}(\overline{WB},\hat{D},M)-Soc(\hat{J}(\overline{WB}_{-n},\hat{D},M)\leq \sum_{j\in J_n}(m^{hhp(j)}\cdot \ell(j)-b(j))\end{equation}
Therefore, due to equation~(\ref{eqPaymentRule}), algorithm~\ref{algValue} and assumption~\ref{eq08012025simplicity}, we have
$$\sum_{n\in N}payment^n(\overline{WB},\hat{D},M)\leq  \sum_{n\in N}\sum_{j\in J_n}(m^{hhp(j)}\cdot \ell(j)-b(j))= Soc(\hat{J}(\overline{WB},\hat{D},M)) . $$

which is a contradiction to equation~(\ref{eq150724b}). 
 \end{proof}

\textbf{Lemma~\ref{lmProbNP}}
Computing the probability of supplying $y$ kWh during $hhp$ given $J^{active}$, i.e, $p_{hhp}(y,J^{active})$,   is an NP-hard problem.

\begin{proof}  

The proof is by reduction from the Knapsack problem~\cite{karp1972complexity}, where we consider a Knapsack instance $\{(v_i,w_i)\}_{i\in N}$ (volume, worth). 

We consider the items in the knapsack problem $\tilde{N}$ equivalent to the set of contracts $J^{active}$ thus: $\forall i\in \tilde{N}:(v_i,w_i)$ where $v_i$ ($w_i$) is the volume (worth) of item $i$: $\exists j\in J^{active}$, where $ln(p(j))=v_i$ (according to lemma~\ref{lmProbEstimation}) and $\ell(j)=w_i$. 
The term $\bar{p}_{hhp}(y, j^{active})$ describes the probability that that no less than $y$ $kWh$ are supplied during $hhp$, and $ln(\bar{p}_{hhp}(y, j^{active}))$  shall be equivalent to the probability that the sack will be filled no less than $Z$. Therefore,  $ln(\bar{p}_{hhp}(y,j^{active}))'=\frac{1}{\bar{p}_{hhp}(y,j^{active})}$ 
is equivalent to the probability of filling a sack of size $Z$.
\end{proof}

\textbf{Proposition~\ref{propBoundForProbabilty} }
Given a set of  contracts $J^{active}$,  the approximation $\hat{p}_{hhp}(y, J^{active})$ (see proposition~\ref{prop20012023b}), is a lower bound of the true probability of  supplying  $y$ $kWh$ at $hhp $ to  the platform, i.e.
$$\hat{p}_{hhp}(y, J^{active})\leq {p}_{hhp}(y, J^{active})$$
where $p_{hhp}(y,J^{active})$ is the true probability of $y$ $kWh$ discharged to the platform via contracts $J^{active}$ at $hhp\in \mathbb{T}$.

\begin{proof}
We wish to show that
\begin{equation}\label{eq90125}
\hat{p}_{hhp}(y, J^{active})=\cfrac{|\frac{{X}^{max}_{hhp}}{\bar{\ell}}|!}{[\frac{y}{\ell}]!\cdot \left[[\frac{{X}^{max}_{hhp}}{\bar{\ell}}]-[\frac{y}{\bar{\ell}}]\right]!}\cdot [(\hat{p}_{hhp})^{\frac{y}{\bar{\ell}}}]\cdot [(1-\bar{p}_{hhp})^{([\frac{{X}_{hhp}^{max}}{\bar{\ell}}]-[\frac{y}{\bar{\ell}}])}]\leq \end{equation}
$$ \leq NumCom \cdot 
\mathbb{E}_{\tilde{J}\subseteq \hat{J}(\overline{WB},\hat{D},M):\sum_{j\in \tilde{J}}\ell(j)=(X^{max}_{hhp}-y)}\Pi_{j\in \tilde{J}}{(1-p(j))} \cdot 
\mathbb{E}_{\tilde{J}\subseteq \hat{J}(\overline{WB},\hat{D},M):\sum_{j\in \tilde{J}}\ell(j)=y}\Pi_{j\in \tilde{J}}({p(j)})$$
 where
$ \ y'={D}^{hhp}-y \ $
  if there is a deficit of $y'$ $kWh$ (and  $ \   y'=y-\hat{D}^{hhp} $ if there is an excess $y'$ $kWh$ supplied), and
 where $\bar{p}$ is in accordance with equation~(\ref{eq04012025}), and
where $NumCom$ is equal to the  number  of  different  combinations  of  honored  and defaulted contracts  that could  cover  demand  y.

For simplicity, we assume that  $NumCom = \cfrac{|\frac{X^{max}_{hhp}}{\bar{\ell}}|!}{[\frac{y}{\ell}]!\cdot \left[[\frac{X^{max}_{hhp}}{\bar{\ell}}]-[\frac{y}{\bar{\ell}}]\right]!}$.

 Due to the Isoperimetric theorem~\cite{pressley2010elementary},
$\forall X^{max}_{hhp}\in \mathbb{R}_{\geq 0}$ and $\forall y\in [0,X^{max}_{hhp}]$, he have:

$$
(\bar{p})^{\frac{y}{\bar{\ell}}}\leq 
\mathbb{E}_{\tilde{J}\subseteq \hat{J}(\overline{WB},\hat{D},M):\sum_{j\in \tilde{J}}\ell(j)=y}\Pi_{j\in \tilde{J}}({p(j)})$$
\begin{center}
and
\end{center} 
$$ (1-\bar{p})^{([\frac{X_{hhp}^{max}}{\bar{\ell}}]-[\frac{y}{\bar{\ell}}])}\leq
 \mathbb{E}_{\tilde{J}\subseteq \hat{J}(\overline{WB},\hat{D},M):\sum_{j\in \tilde{J}}\ell(j)=(X^{max}_{hhp}-y)}\Pi_{j\in \tilde{J}}{(1-p(j))} $$

Therefore 
$$ 
[(\bar{p})^{\frac{y}{\bar{\ell}}}]\cdot [(1-\bar{p})^{([\frac{X_{hhp}^{max}}{\bar{\ell}}]-[\frac{y}{\bar{\ell}}])}]\leq
$$
$$\leq [\mathbb{E}_{\tilde{J}\subseteq \hat{J}(\overline{WB},\hat{D},M):\sum_{j\in \tilde{J}}\ell(j)=y}\Pi_{j\in \tilde{J}}({p(j)})]\cdot [\mathbb{E}_{\tilde{J}\subseteq \hat{J}(\overline{WB},\hat{D},M):\sum_{j\in \tilde{J}}\ell(j)=(X^{max}_{hhp}-y)} \Pi_{j\in \tilde{J}}{(1-p)}].
$$

In addition to the above analysis recall that  $\forall j\in J^{active}: \hat{p}(j)\leq p(j)$ (see lemma~\ref{lmProbEstimation}) which further ascertains that $\hat{p}(y,J^{active})$ is a lower bound to $p(y, J^{active})$.
\end{proof}

\end{appendix}

\end{document}